\documentclass[11pt]{article} 
\usepackage{amsmath,amssymb,cite,graphicx}
\usepackage{hyperref}
\usepackage{url}
\newtheorem{proposition}{Proposition}[section]
\newtheorem{examp}[proposition]{Example}
\newtheorem{fact}[proposition]{Fact}
\newtheorem{lemma}[proposition]{Lemma}

\newtheorem{theorem}[proposition]{Theorem}

\newenvironment{example}{\begin{examp}\rm}{\end{examp}}
\newenvironment{proof}{\par\noindent {\it Proof.} \rm}{}
\newcommand{\qed}{\hfill$\square$}
\newenvironment{proof2}[1]{\par\noindent {\it Proof of #1.} \rm}{ \qed} 
\newcommand{\bw}{\mathbf{w}}
\usepackage{geometry}
\begin{document}

 \pagestyle{plain}

\title{A characterization of binary morphisms\\ generating Lyndon infinite words}
\author{Gwenaël Richomme, Patrice S\'e\'ebold\\
LIRMM, Université Paul-Valéry Montpellier 3,\\
Université de Montpellier, CNRS, Montpellier, France}
\date{\today}

\maketitle

\begin{abstract}
An infinite word is an infinite Lyndon word if
it is smaller, with respect to the lexicographic order, than all its proper
suffixes, or equivalently if it has infinitely many finite Lyndon words as prefixes.
A characterization of binary endomorphisms generating Lyndon infinite words is provided. 
\end{abstract}

\section{Introduction}

Finite Lyndon words are the non-empty
words which are smaller, w.r.t. (with respect to) the lexicographic order,
than all their proper suffixes. They are important tools in many studies (see, \textit{e.g.}, \cite{BerstelLauveReutenauerSaliola2008book,Lothaire1983book,Lothaire2002book,Reutenauer2019book}).
Infinite Lyndon words are defined similarly. They are also the words that have infinitely many finite Lyndon words as prefixes.
They occur in many context (see, \textit{e.g.}, 
\cite{AlloucheCurrieShallit1998,BorelLaubie1993,CharlierKamaePuzyninaZamboni2014_JCTA,GlenleveRichomme2008TCS,LeveRichomme2007TCS,Paquin2010DMTCS,PosticZamboni2020TCS}).

The aim of the current paper is to provide a characterization, in the binary case, 
of endomorphisms that generate infinite Lyndon words. 
This paper continues the study of links between morphisms and Lyndon words done by the first author. 
In \cite{Richomme2003BBMS} he studied and characterized
the morphisms that preserve Lyndon words, calling them \textit{Lyndon morphisms}: these morphisms are
those that map any Lyndon word to another Lyndon word.
This study was extended to morphisms that preserve infinite Lyndon words in \cite{Richomme2007DMTCS}.

Note that being a morphism that preserves finite Lyndon words is a sufficient condition 
to generate an infinite Lyndon word (if the morphism generates an infinite word).
Indeed if $f$ is a morphism that preserves finite Lyndon words and $u$ is a Lyndon word,
then, for any $n \geq 0$, $f^n(u)$ is a Lyndon word. Applying this process when $u = a$ with 
a morphism $f$ that generates from $a$ an infinite word $\bw$, 
we see that $\bw$ has infinitely many finite Lyndon words as prefixes: it is an infinite Lyndon word.
But the condition is not necessary. 
For instance, the morphism defined by $f(a) = aba$ and $f(b) = bb$ generates an infinite Lyndon word (the proof can be done using Proposition~\ref{propCas3}) but it does not preserve finite Lyndon words since $f(a)$ is not a Lyndon word.

Our main characterization is Theorem~\ref{th_main}:
Over $\{a, b\}$ with $a \prec b$, a non-periodic word generated by a morphism $f$ 
prolongable on $a$ is an infinite Lyndon word if and only if 
$f$ preserves the lexicographic order on finite words and $f^3(a)$ is a prefix of Lyndon words.
The proof needs to consider separately the case where $aa$ is a prefix of $f^\omega(a)$ 
and the case where $ab$ is a prefix of $f^\omega(a)$.
After some needed preliminaries in Section \ref{sectionPreliminaries}, we prove the following general necessary condition: 
a binary endomorphism that generates an infinite Lyndon word must preserve the lexicographic order on finite words.
In Section~\ref{sectionCase2}, we characterize morphisms that generate an infinite Lyndon word beginning with $aa$ (Proposition~\ref{propCas2}).
In Section~\ref{sectionCase3}, we characterize morphisms that generate an infinite Lyndon word beginning with $ab$ (Proposition~\ref{propCas3}).
In Section~\ref{Section_main_result}, we prove our mail result.
We conclude with a few words on what happens on larger alphabets.

\section{\label{sectionPreliminaries}About Lyndon words and morphisms}

We assume that readers are familiar with combinatorics on words and morphisms 
(see, \textit{e.g.}, \cite{Lothaire1983book,Lothaire2002book}). 
We specify our notation and recall useful results.

 An \textit{alphabet} $A$ is a set of symbols called \textit{letters}.
 Here we consider only finite alphabets.  A \textit{word over} $A$ is
 a sequence of letters from $A$.  The \textit{empty word}
 $\varepsilon$ is the empty sequence.  Equipped with the concatenation
 operation, the set $A^*$ of finite words over $A$ is a free monoid
 with neutral element $\varepsilon$ and set of generators $A$.  We let $A^\omega$ 
 denote the set of infinite words over $A$.  As usually,
 for a finite word $u$ and an integer $n$, the $n^{\rm th}$ power of
 $u$, denoted $u^n$, is the word $\varepsilon$ if $n = 0$ and the word
 $u^{n-1}u$ otherwise.  If $u$ is not the empty word, $u^\omega$
 denotes the infinite word obtained by infinitely repeating $u$. Such a word is called \textit{periodic}.  A
 finite word $w$ is said \textit{primitive} if for any word $u$, the
 equality $w = u^n$ (with $n$ an integer) implies $n = 1$. 

 Given a non-empty word $u = a_1\cdots a_n$ with $a_i \in A$, the
 \textit{length} $|u|$ of $u$ is the integer $n$. One has
 $|\varepsilon| = 0$.  If for some words $u, v, p, s$ (possibly
 empty), $u = pvs$, then $v$ is a \textit{factor} of $u$, $p$ is a
 \textit{prefix} of $u$ and $s$ is a \textit{suffix} of $u$.  When $p
 \neq u$ (resp. $s \neq u$), we say that $p$ is a \textit{proper
 prefix} (resp. $s$ is a \textit{proper suffix}) of $u$. 

Let us recall two basic results.

\begin{proposition}[{see, \textit{e.g.}, \cite[Prop. 1.3.2]{Lothaire1983book}}]
\label{P1.3.2.Lothaire}
For any words $u$ and $v$, $uv = vu$ if and only if there exist a word $w$ and integers $k$, $\ell$ such that $u = w^k$ and $v = w^\ell$.
\end{proposition}

\begin{theorem}[Fine and Wilf's Theorem, see, \textit{e.g.}, {\cite[Prop. 1.3.5]{Lothaire1983book}}]
\label{Th_Fine_Wilf}~\\
Let $x, y \in A^*$, $n = |x|$, $m = |y|$, $d = gcd(n, m)$.
Assume there exist integers $p$ and $q$ such that
$x^p$ and $y^q$ have a common prefix
of length at least equal to $n+m-d$.
Then $x$ and $y$ are powers of the same word.
\end{theorem}
 
 \subsection{Lyndon words}
 
 From now on we consider ordered alphabets.  We let $A_n = \{a_1 \prec
 \ldots \prec a_n\}$ denote the $n$-letter alphabet $A_n = \{a_1, \ldots, a_n\}$
 with order $a_1 \prec \ldots \prec a_n$. Given an ordered alphabet $A$, we
let also $\preceq$ denote the lexicographic order whenever used on $A^*$ or
 on $A^\omega$.  Let us recall that for two different (finite or
 infinite) words $u$ and $v$, $u \prec v$ if and only if 
 $u = x\alpha y$, $v
 = x\beta z$ with $\alpha, \beta \in A$, $\alpha  \prec \beta$, $x \in A^*$, $y, z \in A^*\cup
 A^\omega$, or if (when $u$ is finite) $u$ is a proper prefix of $v$. 
For any finite words $u$, $v$, $w$, if $u \prec v$, then $wu \prec wv$.
Moreover if $u$ is not a prefix of $v$ and $u \prec v $, then
$ux \prec vy$ for any words $x$ and $y$.

 A non-empty finite word $w$ is a \textit{Lyndon word} if for all
non-empty words $u$ and $v$, $w = uv$ implies $w \prec vu$.  
Equivalently \cite{ChenFoxLyndon1958,Lothaire1983book}, a non-empty word $w$ is
a Lyndon word if all its non-empty proper suffixes are greater than itself
for the lexicographic order. For instance, on the one-letter alphabet
$\{a\}$, only $a$ is a Lyndon word.  On $\{a \prec b\}$ the Lyndon words
of length 6 are 
$aaaaab$, 
$aaaabb$, $aaabab$, 
$aaabbb$, $aababb$, $aabbab$,
$aabbbb$, $ababbb$,
$abbbbb$.
Lyndon words are
primitive.  
Note that Lyndon words have no non-empty border, that is, 
there is no proper prefix of a Lyndon word $u$ that is also a suffix of $u$.
Observe also that if $u$ is a prefix of a Lyndon word
then there cannot exist words $v$ and $w$ 
such that the three following conditions hold:
$v$ is a prefix of $u$;
$w$ is a factor of $u$ which is not
a prefix of $u$;
$w \prec v$.

\begin{proposition}[{see, \textit{e.g.}, \cite[prop. 5.1.3]{Lothaire1983book}}]
\label{baseLyndon}
A non-empty word $w$ is a Lyndon word if and only if $|w| = 1$ or $w =
uv$ with $u$ and $v$ two Lyndon words such that $u \prec v$.
\end{proposition}

Lyndon infinite words were introduced in \cite{SiromoneyMathewDareSubramanian1994} as the infinite words that have 
infinitely many prefixes that are Lyndon words. It follows from the definition that an infinite Lyndon word is not periodic.
More generally an infinite word is Lyndon if and only if all its proper suffixes
are greater than it w.r.t. the lexicographic order \cite[Prop. 2.2]{SiromoneyMathewDareSubramanian1994}. 

\subsection{Morphisms}

Let $A$ and $B$ be two alphabets. 
A \emph{morphism} $f$
from $A^*$ to
 $B^*$ is a mapping from $A^*$
 to $B^*$ such that for all words $u,
v$ over $A$, $f(uv) =
f(u)f(v)$. We say that $f$ is a morphism over $A$ if we don't need to refer to $B$.
When $A = B$, $f$ is an \emph{endomorphism} over $A$.
A morphism is \textit{erasing} if $f(a) = \varepsilon$ for some letter $a$.
For $n \geq 0$ and any word (finite or infinite) $u$, $f^n(u)$ is $u$ if $n = 0$ and $f^{n-1}(f(u))$ otherwise.

An endomorphism is said \textit{prolongable} on $a$ if $f(a) = au$ for some word $u$ and if $\lim_{n \to \infty} |f^n(a)|=\infty$.
For such a morphism, for all $n \geq 0$, $f^n(a)$ is a prefix of $f^{n+1}(a)$. Then the sequence $(f^n(a))_{n \geq 0}$ defines a unique infinite word, denoted $f^\omega(a)$. This word is a fixed point of $f$.

A morphism \textit{preserves finite Lyndon words} if and only if the image of any finite Lyndon word is also a Lyndon word.
Similarly morphisms that \textit{preserve infinite Lyndon words} can be defined.
A morphism \textit{preserves the order on finite words} if, for all words $u$ and $v$,
$u \prec v$ implies $f(u) \prec f(v)$. Such a morphism is injective and so non-erasing.
In \cite{Richomme2003BBMS}, it is proved that a morphism is a
Lyndon morphism if and only if it preserves the lexicographic order on
finite words and if the image of each letter is a Lyndon word. 
We have also the following characterization.
 
\begin{proposition}[{\cite[prop. 3.3]{Richomme2003BBMS}}]
\label{Richomme2003BBMS_prop3.3binaire}
A morphism $f$ over $\{a \prec b\}$ preserves the lexicographic order on finite words
if and only if $f(ab) \prec f(b)$.
\end{proposition}

\section{\label{preliminaireEngendrement}A necessary condition}

In this section we prove the following result that states a necessary condition for a prolongable binary morphism to generate an infinite Lyndon word.

\begin{proposition}
\label{engendreEntrainePresMotsFinis}
Let $f$ be an endomorphism over $\{a \prec b\}$. 
Assume that $f$ is prolongable on $a$.
If $ f^\omega(a)$ is a Lyndon infinite word
then $f$ preserves the lexicographic order on finite words.
\end{proposition}

We will use the basic fact and the following characterization of prefixes of Lyndon words.
\begin{fact}
\label{fait_Preliminaire}
Given any finite Lyndon word $x$ and any proper non-empty prefix $p$ of $x$,
$px \prec x$. 
\end{fact}

\begin{proof}
Let $q$ be the word such that $x = pq$.
Since $x$ is a Lyndon word and since $x \neq q$ and $x \neq \varepsilon$, 
$x \prec q$. It follows that $px \prec pq = x$.
\qed
\end{proof}

\begin{proposition}[{\cite[Prop. 1.7]{Duval1983JA}}]
\label{prefix_Lyndon_words}
Let $A$ be an ordered alphabet with maximal letter $c$.
Let $P$ be the set of prefixes of Lyndon words. The set $P \cup \{c^k \mid k \geq 2\}$ is equal to the set 
of all words on the form $(uv)^ku$ with $k \geq 1$ an integer and $u$, $v$ some finite words such that 
$v \neq \varepsilon$ and $uv$ is a Lyndon word.
\end{proposition}

\begin{proof2}{proposition~{\ref{engendreEntrainePresMotsFinis}}}

Assume by contradiction that $f$ does not preserve
the lexicographic order on finite words. 
By Proposition~\ref{Richomme2003BBMS_prop3.3binaire},
$f(b) \prec f(ab)$ (the equality cannot hold as $f(a)$ is not empty). 
Thus,
for any integer $n \geq 0$, $f(a^nb) \prec f(a^{n+1}b)$.
So, for any integer $n \geq 0$, $f(b) \prec f(a^nb)$. 

From now on let $i$ be the integer such that $a^ib$ is a prefix of $f^\omega(a)$.
Let also $\bw$ be the word such that $f^\omega(a) = f(a^ib) \bw$. Note that $i \geq 1$
since $b^\omega$ is not an infinite Lyndon word.

Observe that $f(b)$ is a prefix of $f(a^ib)$. 
Otherwise, from $f(b) \prec f(a^ib)$, 
we deduce that $f(b)\bw \prec f^\omega(a)$ which contradicts the fact that $f^\omega(a)$ 
is an infinite Lyndon word since $f(b)\bw$ is a proper suffix of $f^\omega(a)$.

As $f^\omega(a)$ is an infinite Lyndon word, it has infinitely many prefixes that are Lyndon words.
Thus its prefix $f(a^ib)$ is a prefix of a Lyndon word. Hence by Proposition~\ref{prefix_Lyndon_words}, 
there exist an integer $k \geq 1$ and
words $u$ and $v$ 
such that $f(a^ib) = (uv)^ku$, $v \neq \varepsilon$ and $uv$ is a Lyndon word.
Consequently $f(b) = (uv)^ju'$ for some $j \geq 0$ and some proper prefix $u'$ of $uv$.

Observe that $ba$ is a factor of $f^\omega(a)$.
Indeed otherwise $f^\omega(a) = a^ib^\omega$ which implies 
$i = 1$ and $f(b) \in b^+$, and so, a contradiction with $f(b) \prec f(ab)$.

Assume that $u' \neq \varepsilon$.
Since $ba$ is a factor of $f^\omega(a)$, 
the word $u'uv$ is a factor of $f(ba)$ and so of $f^\omega(a)$.
By Fact ~\ref{fait_Preliminaire}, $u'uv \prec uv$: 
since $uv$ is a prefix of $f^\omega(a)$, this contradicts
the fact that $f^\omega(a)$ is an infinite Lyndon word.

Thus $u' = \varepsilon$. 
This means that $f(b) = (uv)^j$ with $j \geq 0$.
If $j = 0$, $f(b) = \varepsilon$ and 
$f^\omega(a) = f(a)^\omega$ is a periodic word: a contradiction with the fact it is an infinite Lyndon word.
Thus $j \geq 1$.
Since $f(b)$ is a suffix of $f(a^ib) =
(uv)^ku$, we get $uv = vu$. 
Remember that $v \neq \varepsilon$. 
If $u \neq \varepsilon$, by Proposition~\ref{P1.3.2.Lothaire}, the word
$uv$ is not primitive: a contradiction with the 
primitivty of the Lyndon word $uv$.
So $u = \varepsilon$.  

This implies that both $f(a)$ and $f(b)$ are powers of
$v$. So $f^\omega(a) = v^\omega$. This is a final contradiction with the fact that
an infinite Lyndon word cannot be periodic. The morphism $f$ preserves the order on finite words over $\{a \prec b\}$.
\end{proof2}

\medskip

Note that the converse of Proposition~\ref{engendreEntrainePresMotsFinis} does not hold. Consider, for instance, the morphism $f$ defined by $f(a) = abb$ and $f(b) = baa$. This morphism preserves the lexicographic order on infinite word 
but the word $f^\omega(a)$ is not an infinite Lyndon word.

\medskip

One could expect a stronger necessary condition as, for instance, a preservation of infinite Lyndon words.
The next example shows that this stronger condition is not necessary.

Let $f$ be defined by $f(a) = aab$ and $f(b) = abaabab$.
The word $\bw = abbabbb^\omega$ is an infinite Lyndon word.
Its image by $f$ begins with 
$ubua$ where $u = aababaaba$.
Hence $f$ does not preserve infinite Lyndon words.
Nevertheless using Proposition~\ref{propCas2}, one can verify that $f$ generates an infinite Lyndon word.

\section{\label{sectionCase2}Generating infinite Lyndon words beginning with aa}

We consider here the case of generated words beginning with $aa$.

\begin{proposition}
\label{propCas2}
Let  $f$ be an endomorphism over $\{a\prec b\}$ prolongable on $a$ 
such that
$f^\omega(a)$ begins with $a^ib$ for some integer $i \geq 2$. 

The word 
$f^\omega(a)$ is an infinite Lyndon word
if and only if
\begin{enumerate}
\item $f$ preserves the lexicographic order on finite words, and,
\item $f(a^ib)$ is a Lyndon word.
\end{enumerate}
\end{proposition}

The proof of this proposition is based on the next lemmas.

\begin{lemma}
\label{xx2}
Let $f$ be a morphism that preserves the order on finite words.
Let $i \geq 2$. Assume that $f(a^ib)$ is a Lyndon word.
For any word $v$ such that 
$a^ibv$ is a Lyndon word, 
the word $f(a^ibv)$ is also a Lyndon word.
\end{lemma}

\begin{proof}
We act by induction on $|v|$.

By hypothesis the result holds when $|v| = 0$.
Assume that $|v| \geq 1$. 
By Proposition~\ref{baseLyndon}, 
there exist Lyndon words $\ell$ and $m$ 
such that $a^ibv = \ell m$ and $\ell \prec m$.
Let us choose $m$ with the smallest  length as possible.

Let us prove that $a^ib$ is a prefix of $\ell$.
Assume that this does not hold.
Then $\ell = a$ and $m =a^{i-1}bv$. 
Consequently, since $m$ is a Lyndon word, $a^i$ is not a factor of $a^{i-1}bv$. 
Let $m'$ be the suffix of $m$ such that $bm' \in ba^+b^+$.
If such a factor does not exist (that is if $m \in a^{i-1}b^+$), let $m' = b$. 
In all cases, $m'$ is a Lyndon word. 
Let $\ell'$ be the word such that $\ell m = \ell'm'$. 
The word $a^ib$ is a prefix of $\ell'$ (when $m =b$, remember that $|v| \geq 1$).
Observe that $\ell' \prec m'$.
The last letter of $\ell'$ is $b$. Indeed, by construction, it could be the letter $a$ only if 
$m \in a^{i-1}b^+$, that is if $m = a^{i-1}b^k$ for some $k \geq 1$. But then $m' = a$
and $\ell' = a^{i-1}b^{k-1}$. As $|v| \geq 1$, we have $k \geq 2$, and so, the last letter of $\ell'$ is $b$.
Let $s$ be a proper non-empty suffix of $\ell'$. 
Let $j \geq 0$ be the integer such that $s$ begins with $a^jb$.
Since $a^ib$ is not a factor $a^{i-1}bv$ and since $a^{i-1}bv$ is a Lyndon word, 
we deduce that $j < i$. 
So $\ell' \prec s$. Hence $\ell'$ is a Lyndon word: this contradicts the choice made on $m$ and proves that $a^ib$ is a prefix of $\ell$.

If $a^ib$ is a prefix of $m$ then, by inductive hypothesis,
$f(\ell)$ and $f(m)$ are Lyndon words.
Since $\ell \prec m$ and $f$ preserves the order on finite words, 
$f(\ell) \prec f(m)$. Proposition~\ref{baseLyndon} implies that 
$f(a^ibv) = f(\ell m)$ is a Lyndon word.

From now on assume that $a^ib$ is not a prefix of $m$.
Observe that this implies that $a^i$ is not a factor of $m$. 
Indeed since $a^ibv = \ell m$ is a Lyndon word, 
for any factor $a^j$ of $\ell m$, we have $j \leq i$.
So $m$ begins with $a^kb$ for some integer $k < i$. 
Moreover as $m$ is a Lyndon word, for any factor $a^j$ of $m$, we have $j \leq k < i$.
Let $s$ be a proper non-empty suffix of $\ell m$.  
If $|s| \leq |f(m)|$ then 
there exist an integer $j < i$,
a word $m'$ 
and a non-empty suffix $s'$ of
$f(a^jb)$
such that 
$s = s'f(m')$.  
The word $s'$ is a proper non-empty suffix of the Lyndon word
$f(a^ib)$. So $f(a^ib) \prec s'$ and $f(a^ibv) \prec s' \preceq s$.
If $|s| > |f(m)|$ then $s = s'f(m)$ with $s'$ a proper non-empty suffix of $f(\ell)$.
By inductive hypothesis, $f(\ell)$ is a Lyndon word.
Thus $f(\ell) \prec s'$ and consequently $f(a^ibv) = f(\ell m) \prec s' \prec s'f(m) = s$. 
The word $f(a^ibv)$ is a Lyndon word.
\qed\end{proof}

\begin{lemma}
\label{U1}
Let $u$ be a non-empty word.
If $uu$ is a prefix of a Lyndon word, 
then $u$ is a power of a Lyndon word.
\end{lemma}

\begin{proof}
Since $uu$ is a prefix of a Lyndon word, 
also $u$ is a prefix of this Lyndon word.
By Proposition~\ref{prefix_Lyndon_words},
there exist words $x$ and $y$ such that $y \neq \varepsilon$,
$xy$ is a Lyndon word and for some integer $k \geq 1$, $u = (xy)^kx$.
If $x \neq \varepsilon$, since $xy$ is a Lyndon word,
we have $xy \prec y$ and so $xxy \prec xy$.
Then for any word $v$,  
the word $x(xy)^kxv$ is a suffix of $uuv$ 
and $x(xy)^kxv \prec uuv$. 
This contradicts the fact that $uu$ is a prefix of a Lyndon word.
So $x = \varepsilon$. This implies that $u = y^k$ and $y$ is a Lyndon word.
\qed\end{proof}

\begin{lemma}
\label{suff_aa}
Assume that $f$ is an endomorphism over $\{a \prec b\}$ prolongable on $a$
such that $f^3(a)$ is a prefix of a Lyndon word, 
$f^\omega(a)$ begins with the word $a^ib$ for some integer $i \geq 2$
and $f^\omega(a)$ is not periodic.
Then $f(a^ib)$ is a Lyndon word.
\end{lemma}

\begin{proof}
Let us first observe that $f(a)$ begins with $a^i b$. Indeed otherwise $f(a)$ is a power of $a$ contradicting the non-periodicity of $f^\omega(a)$. 

Observe also that the word $f(a^ib)$ is a prefix of $f^2(a)$  which itself is a prefix of $f^3(a)$.
Hence $f(a^ib)$ is also a prefix of a Lyndon word.
By Proposition~\ref{prefix_Lyndon_words}, 
there exists a Lyndon word $v$, a proper prefix $p$ of $v$ ($p$ may be empty)
and an integer $\ell \geq 1$ such that $f(a^ib) = v^\ell p$.
Since $i \geq 2$, by Lemma~\ref{U1}, $f(a)$ is a power of a Lyndon word $u$.

If $v = u$, from  $f(a^ib) = v^\ell p$, we get $f(b) = v^{\ell'} p$ for some integer $\ell'$. 
In particular $p$ is a suffix of $f(b)$. 
If $p = \varepsilon$, then $f^\omega(a) = v^\omega$ a contradiction with its non-periodicity.
Assume now that $p \neq \varepsilon$.
Since $i \geq 2$, the word $a^ib$ occurs twice in $f(a^ib)$ which is a prefix of $f^2(a)$. 
Thus $ba$ is a factor of $f^2(a)$ and $f(ba)$ is a factor of $f^3(a)$.
Then the word $pu$ is a factor of $f^3(a)$. 
Note also that $u$ is a prefix of $f^3(a)$. 
As $p$ is a proper non-empty prefix of the Lyndon word $u$, 
by Fact~\ref{fait_Preliminaire}, $pu \prec u$.
This contradicts the fact that $f^3(a)$ is a prefix of a Lyndon word.
Thus $v\neq u$.

Since $i \geq 2$, $u^2$ is a prefix of $v^{\ell+1}$.
If $|u| \geq |v|$, by Theorem~\ref{Th_Fine_Wilf}, 
$u$ and $v$ are powers of the same word. This is not possible as $u \neq v$ and 
both words $u$ and $v$ are primitive (since they are Lyndon words).
Thus $|v| > |u|$. 

Note that $v$ is not a factor of $f(a^i) = f(a)^i$. 
Indeed if $v$ is a factor of $f(a)^i$ then it is a prefix of a power of $u$, and so, a prefix of $u$ is both a prefix and a suffix of $v$: this is impossible since $v$ is a Lyndon word.
It follows  that $p$ is a proper suffix of $f(b)$.

Observe that $a^ib$ is a prefix of $f(a)$ and so $f(a)a^ib$ is a prefix of $f(aa)$ and so of $f^2(a)$.
Since $f^2(a)$ is a prefix of a Lyndon word, it cannot contain $a^{i+1}$ as a factor
and so the last letter of $f(a)$ must be $b$. 
Hence $ba^ib$ and $f(ba^ib)$ are  factors of $f^3(a)$. 
This implies that $pv$ is also a factor of 
$f^3(a)$. 
By Fact~\ref{fait_Preliminaire}, $pv \prec v$ if $p \neq \varepsilon$. 
This contradicts the fact that $f^3(a)$ is a prefix of a Lyndon word.
So $p =\varepsilon$ and $f(a^ib) = v^\ell$.
Assume that $\ell \geq 2$.

Since $a^ib$ is a prefix of $f(a)$
and since $f(a^ib) = v^\ell$, 
$a^ib$  is also a prefix of $v$.
Thus $v^\ell$ is a prefix of $f(v)$ itself a prefix of $f^2(a)$.
Since $\ell \geq 2$, $f(v)f(v)$ and $f(v)v^\ell$
are prefixes of $f^3(a)$. 

Let us prove that $f(v)$ is not a prefix of $v^\omega$.
Assume by contradiction that $f(v) = v^kp'$
for some proper prefix $p'$ of $f(v)$ and
some integer $k$.
Since $v^\ell$ is a prefix of $f(v)$,
we have $k \geq \ell \geq 2$.
If $p' \neq \varepsilon$, by Fact~\ref{fait_Preliminaire}, $p'v \prec v$.
Since $p'v$ is a factor of $f(v)f(v)$,
this contradicts the fact that $f^3(a)$ is a prefix of a Lyndon word.
So $p' = \varepsilon$ and $f(v) = v^k$.
Hence by induction, for all $n \geq 0$,
$f^n(v) \in v^+$.  Moreover we have $\lim_{n \to \infty} |f^n(v)| = \infty$.
So $f^\omega(a) = v^\omega$: a contradiction with the non-periodicity of $f^\omega(a)$.

So $f(v)$ is not a prefix of $v^\omega$.
There exists an integer $k$,
a proper prefix $\pi$ of $v$
and letters $\alpha$, $\beta$
such that $v^k\pi\beta$ is a prefix of $f(v)$ and $\pi\alpha$ is a prefix of $v$.
Since $f^3(a)$ is a prefix of a Lyndon word,
$\alpha = a$ and $\beta = b$.
Note that $v^{k+1} \prec v^k\pi\beta$.

We have already mentioned that $v$ is not a factor of $f(a^i)$.
From $f(a^ib) = v^\ell$ and $\ell \geq 2$, 
we deduce that $v$ is a suffix of $f(b)$.
Moreover since $v$ is a Lyndon word beginning with $a^ib$, the last letter of $v$ is $b$:
$v$ is so a suffix of $f(v)$.
Since $f(v)f(v)$ is a factor of $f^3(a)$,
the word $v^{k+1}$ is a factor of $f^3(a)$,
This contradicts the fact that $f^3(a)$ is a prefix of a Lyndon word.

Thus $\ell = 1$: $f(a^ib)$ is a Lyndon word.
\qed\end{proof}

\begin{proof2}{Proposition~{\ref{propCas2}}}
We first prove that the two conditions are sufficient.
First observe that, for any integer $n \geq 1$, $a^ib$ is a prefix of $f^n(a^ib)$ 
(this is a direct consequence of the facts that $f$ is prolongable on $a$
and that $f^\omega(a)$ begins with $a^ib$).
Thus by induction, using Lemma~\ref{xx2},  we get:
for any integer $n \geq 0$, $f^n(a^ib)$ is a Lyndon word. 
As $\lim_{n \to \infty} |f^n(a^ib)| = \infty$, 
the word $f^\omega(a)$ has infinitely many prefixes that are Lyndon words.
By definition, it is an infinite Lyndon word.

From now on 
assume that $f^\omega(a)$ is an infinite Lyndon word.
Proposition~\ref{engendreEntrainePresMotsFinis} shows that $f$
preserves the lexicographic order on finite words.
Observe that since it is an infinite Lyndon word, $f^\omega(a)$ is not periodic 
and $f^3(a)$ is a prefix of a Lyndon word. 
Lemma~\ref{suff_aa} states that $f(a^i b)$ is a Lyndon word.
\end{proof2}

\section{\label{sectionCase3}Generating infinite Lyndon words starting with ab}

We consider here the case of generated words beginning with $ab$.
The word $ab^\omega$ is an infinite Lyndon word. A morphism $f$ generates it if and only if
$f(a) = ab^i$ for some integer $i \geq 1$ and if $f(b) \in b^+$.

In what follows we only consider the case where $f^\omega(a)$ begins with $ab^ia$
for some  $i \geq 1$.

\begin{proposition}
\label{propCas3}
Let $f$ be an endomorphism over $\{a\prec b\}$ prolongable en $a$ 
such that $f^\omega(a)$ begins with $ab^ia$ for some integer $i \geq 1$.

The word $f^\omega(a)$ is an infinite Lyndon word if and only if
\begin{enumerate}
\item $f$ preserves the lexicographic order on finite words,
\item $f(ab^i)$ is a power of a Lyndon word $u \neq ab^i$, and,
\item if $i = 1$, $|u| > |f(b^i)|$.
\end{enumerate}
\end{proposition}

Here follows an example showing that indeed in item 2, $f(ab^i)$ is not necessarily a Lyndon word.

\begin{example}
Let $f$ be defined by $f(a) = abbab$ and $f(b) = b$: $f(ab) = (abb)^2$ is the square of a Lyndon word.
By induction one can verify that all words $f^n(abb)$ are Lyndon words with the relation
$f^{n+1}(abb) = f^n(abb) f^n(abb) b$. This confirms that  $f^\omega(a)$ is an infinite Lyndon word.
\end{example}

We now provide an example showing the necessity of item 3.

\begin{example}
Let $f$ be defined by $f(a) = aba$ and $f(b) = bbababb$: 
$f(ab) = u^2$ with $u = ababb$ is the square of a Lyndon word.
Condition 3 is not verified and indeed $f^\omega(a)$ is not a Lyndon word.
It could be verified that $f^\omega(a)$ begins with $u^4bbu^5$ 
and so contains the factor $u^4a$ which is smaller than the prefix $u^4b$.
\end{example}

The proof of  Proposition~\ref{propCas3} is based on the next lemmas.

\begin{lemma}
\label{I0}
Let $f$ be a morphism that preserves the lexicographic order on finite words over $\{a \prec b\}$. 
Assume that $i \geq 2$ is an integer and that $f(ab^i)$ is a power of a Lyndon word $u$.
Then $|u| > |f(b)|$.
\end{lemma}

\begin{proof}
Assume by contradiction that $|u| \leq |f(b)|$.
Assume first that $|u| < |f(b)|$.
Since $f(b)f(b)$ is a suffix of $f(ab^i)$ so of a power of $u$,
there exist words $p$ and $s$ and an integer $k \geq 1$ 
such that $u = ps$, $f(b)$ ends with $p$ and $f(b) = su^k$.
Since $u$ is a Lyndon word, $p$ cannot be both a prefix and a suffix of $u$ except if $p = \varepsilon$.
When $p = \varepsilon$, $f(b)$ is a power of $u$.
If $|u| = |f(b)|$  then $f(b) = u$.
In all cases both $f(a)$ and $f(b)$ are powers of $u$. 
Hence $f$ is not injective, a contradiction with the fact that 
$f$ preserves the lexicographic order on finite words.
\qed\end{proof}

\begin{lemma}
\label{I1}
Let $f$ be a morphism that preserves the lexicographic order on finite words over $\{a \prec b\}$.  
Assume that $f(ab^i)$ is a power of a Lyndon word $u$ 
for some integer $i \geq 1$. Assume also that $|u| > |f(b)|$ if $i = 1$.
Then, for any non-empty word $v$ over
$\{a \prec b\}$ such that $ab^iv$ is a Lyndon word, 
the word $f(ab^iv)$ is also a Lyndon word.
\end{lemma}

\begin{proof}
We act by induction on $|v|$.  
Let us observe that $f$ is non-erasing and injective
since it preserves the lexicographic order on finite words.
Let $n$ be the integer such that $f(ab^i) = u^n$.
Observe that $|f(b)| < |u|$ (by hypothesis if $i = 1$ and by Lemma~\ref{I0} if $i \geq 2$).

We first assume that $|v| = 1$. 
In this case, since 
$ab^iv$ is a Lyndon word, $v = b$. 
Any suffix $s$ of $f(ab^{i+1})$ with $|s| \leq |f(b)|$
is also a suffix of the Lyndon word $u$.
Thus $u \prec s$ (and for length reason, $u$ is not a prefix of $s$). 
Hence $f(ab^{i+1}) \prec s$. 
Consider now a suffix $s$ of $f(ab^{i+1})$ 
such that $|f(b)| < |s| < |f(ab^{i+1})|$. 
We have $s = s'f(b)$ for some suffix $s'$ of $f(ab^i) = u^n$.
If $s' = s''u^k$ for some proper non-empty suffix $s''$ of $u$ and some integer $k$
then $u \prec s''$ and $u$ is not a prefix of $s''$. Once again $f(ab^{i+1}) \prec s$.
If $s' = u^k$ for some integer $k$ such that 
$1 \leq k < n$, $s = u^k f(b)$.
As $f(b)$ is a proper non-empty suffix of $u$,
$u \prec f(b)$. 
Hence $u^{k+1} \prec u^kf(b)$.
Moreover since $k+1 \leq n$, $f(ab^i) \prec u^k f(b)$.
So for any proper non-empty suffix $s$ of $f(ab^{i+1})$,
$f(ab^{i+1}) \prec s$: $f(ab^{i+1})$ is a Lyndon word.

From now on assume that $|v| \geq 2$.
By Proposition~\ref{baseLyndon}, there
exist two Lyndon words $\ell$ and $m$ 
such that $ab^iv = \ell m$ and $\ell \prec m$. Two cases can hold.
\begin{description}
\item{Case $|m| \geq 2$.} 
As $m$ cannot begin with the letter $b$ (as any Lyndon word of length at least 2 over a binary alphabet),
$\ell$ must begin with $ab^i$.  
Moreover as $ab^iv = \ell m$ is a Lyndon word, $m$ is on the form $ab^k$ with $k > i$ or 
begins with a factor $ab^ka$ with $k \geq i$. 
In both cases, $ab^i$ is a proper prefix of $m$,
and by inductive hypothesis $f(m)$ is a Lyndon word.
If $\ell \neq ab^i$, $f(\ell)$ is also a Lyndon word.
Moreover, since $f$ preserves the lexicographic order, 
$f(\ell) \prec f(m)$.  
By Proposition~\ref{baseLyndon}, 
$f(ab^iv) = f(\ell m)$ is a Lyndon word.
If $\ell = ab^i$, $f(\ell)=u^n$.
Since $f$ preserves the lexicographical order, $u \preceq f(\ell) \prec f(m)$.
Using Proposition~\ref{baseLyndon}, one can prove by induction
that $u^kf(m)$ is a Lyndon word for any $k \geq m$.
Once again, $f(ab^iv) = u ^n f(m)$ is a Lyndon word.

\item{Case $|m| = 1$.} 
In this case, $m = b$. 
Let $s$ be a proper non-empty suffix of $f(ab^iv)$. 
If $|s|  \leq |f(b)|$, 
then $s$ is a suffix of the Lyndon word $u$ (remember that $|f(b)| < |u|$
and $f(ab^i) = u^n$).
This implies that $u \prec s$ and so that $f(ab^iv) \prec s$.
If $|f(b)| < |s| < |f(ab^iv)|$, we have
$s = s'f(b)$ for some proper non-empty suffix $s'$ of the Lyndon word $f(\ell)$ 
(since $|v| \geq 2$, $|\ell| > |ab^i|$ 
and the inductive hypothesis can be applied).
Thus $f(\ell) \prec s'$ which implies that $f(ab^iv) \prec s$.
Hence $f(ab^iv)$ is a Lyndon word.
\end{description}
\qed\end{proof}

\begin{lemma}
\label{suff_ab}
Assume that $f$ is an endomorphism over $\{a \prec b\}$ prolongable on $a$ 
such that $f^3(a)$ is a prefix of a Lyndon word, 
$f^\omega(a)$ begins with the word $ab^ia$ for some integer $i \geq 1$
and $f^\omega(a)$ is not periodic.
Then $f(ab^i)$ is a power of a Lyndon word $u \neq ab^i$.
Moreover if $i = 1$, $|u| > |f(b)|$.
\end{lemma}

\begin{proof}
The word $ab^ia$ is a prefix of $f^\omega(a)$.
Let us prove that the word $f^3(a)$ has a prefix on the form $ab^iab^ka$.
If $f(a)$ has $ab^ia$ as a prefix, 
then $f^2(a)$ (and so $f^3(a)$) contains at least 4 occurrences of $a$.
Since $f^2(a)$ is a prefix of a Lyndon word, it cannot contain the factor $aa$. Hence we get the result.
Assume now that $f(a) = ab^j$ for some $j < i$.
Since $f$ is prolongable on $a$, $j > 0$.
It follows that $f(b)$ begins with $b^{j-i}a$.
Then $f^3(a)$ contains at least 3 occurrences of $a$. 
And once again $f^3(a)$ has a prefix on the form $ab^iab^ka$.

Since $f^3(a)$ is a prefix of a Lyndon word, we have $k \geq i$ and so $(ab^i)^2$ is a prefix of $f^3(a)$.
Lemma~\ref{U1} shows that $f(ab^i)$ is a power of a Lyndon word $u$: 
$f(ab^i) = u^n$ for an integer $n \geq 1$.
If $u = ab^i$, we have $f^\omega(a) = (ab^i)^\omega$ 
which contradicts the fact that $f^\omega(a)$ is aperiodic.
Thus $u \neq ab^i$.

Assume now that $i = 1$ and $|f(b)| \geq |u|$.
From $f(ab) = u^n$ and $f(a) \neq \varepsilon$, we get $n \geq 2$.
Let $s$ be the proper suffix of $u$ 
and let $j \geq 1$ be the integer such that $f(b) = su^j$.
If $s = \varepsilon$, then both $f(a)$ and $(b)$ are powers of $u$. 
This implies that $f^\omega(a) = u^\omega$, a contradiction.
Assume now that $s \neq \varepsilon$.
Let $p$ be the word such that $u = ps$: 
$f(a) = u^kp$ for some integer $k \geq 0$ and $p \not\in \{\varepsilon, u\}$ since $s \not\in \{\varepsilon, u\}$.
Since $u$ is a Lyndon word different from $ab$ but beginning with $aba$, 
we deduce that $u$ begins with $(ab)^mb$ for some $m\geq 2$.

Since $n \geq 2$,
the word $(ab)^mb$ 
has at least one non prefix occurrence in $f(ab)$ so in $f^3(a)$.
This occurrence must be preceded by the letter $b$
since $aa$ cannot occur in $f^3(a)$ which is a prefix of a Lyndon word.
Hence the word $f(ab)^ms =u^{nm}s$ is a prefix of
$f((ab)^mb)$itself a prefix of $f^3(a)$,
and,
the word $uu^{nm}$ which is a suffix of $f(b(ab)^m)$
is a factor of $f^3(a)$.
Since $u \prec s$, we have $u^{nm}u \prec u^{nm}s$:
this contradicts the fact that $f^3(a)$ is a prefix of a Lyndon word.
\qed\end{proof}

\begin{proof2}{Proposition~{\ref{propCas3}}}
Let us fist show that the three conditions imply that $f^\omega(a)$ is an infinite Lyndon word.
Since $f$ preserves the order on finite words, $f$ is not erasing.
Let $u$  be the word occurring in condition 2
and let $k$ be the integer such that $f(ab^i) = u^k$.
Observe that $ab^i$ and the prefix $u$ of $f(ab^i)$ are both prefixes of $f(a)^\omega$.
From $|f(ab^i)| \geq |ab^i|$, $ab^i$ is a prefix of $f(ab^i) =u^k$.
Hence $ab^i $ is a prefix of $u$.
By hypothesis, we cannot have $ab^i = u$. 
So $ab^i$ is a proper prefix of $u$.
For any $n \geq 0$, $f^n(u)$ is a prefix of $f^\omega(a)$ and so 
$ab^i$ is a proper prefix of $f^n(u)$.
Due to condition 3, one can apply Lemma~\ref{I1}.
Thus it follows by induction that $f^n(u)$ is a  Lyndon word for all $n \geq 0$: 
$f^\omega(a)$ is an infinite Lyndon word.

Let us show that the conditions are necessary.
First Proposition~\ref{engendreEntrainePresMotsFinis} 
shows that $f$ preserves the lexicographic order on finite words.
Observe that since it is an infinite Lyndon word, $f^\omega(a)$ is not periodic 
and $f^3(a)$ is a prefix of a Lyndon word. 
Lemma~\ref{suff_ab} states that $f(ab^i)$ is a power of a Lyndon word $u$
and, when $i =1$, $|u| > |f(b)|$.
\end{proof2}

\section{\label{Section_main_result}A general characterization}

Let us prove our main characterization.

\begin{theorem}
\label{th_main}
Let $f$ be an endomorphism over $\{a\prec b\}$ prolongable en $a$.
The word
$f^\omega(a)$ is an infinite Lyndon word 
if and only if
\begin{enumerate}
\item $f$ preserves the lexicographic order on finite words, 
\item $f^\omega(a)$ is not periodic and 
\item the word $f^{3}(a)$ is a prefix of a Lyndon word.
\end{enumerate}
\end{theorem}

\begin{proof}
Assume first that $f^\omega(a)$ is an infinite Lyndon word. 
Conditions 2 and 3 are direct consequences of this hypothesis.
Proposition~\ref{engendreEntrainePresMotsFinis} states condition~1.

Assume now that the three conditions hold.
If $f^\omega(a)$ begins with $aa$, then it begins with $a^ib$ for some integer $i \geq 2$.
Lemma~\ref{suff_aa} states that $f(a^ib)$ is a Lyndon word. 
Thus from Proposition~\ref{propCas2} $f^\omega(a)$ is an infinite Lyndon word.

If $f^\omega(a) = ab^\omega$, it  is an infinite Lyndon word.

If $f^\omega(a)$ begins with $ab^ia$ for some integer $i \geq 1$,
Lemma~\ref{suff_ab} states that $f(ab^i)$ is a power of a Lyndon word $u \neq ab^i$.
Moreover if $i = 1$ then $|u| > |f(b)|$.
Thus from Proposition~\ref{propCas3} $f^\omega(a)$ is an infinite Lyndon word.
\qed\end{proof}

\begin{example}
Let $f$ be the morphism defined by $f(a) = aba$ and $f(b) = bab$.
This morphism fulfills conditions 1 and 3 of Theorem~\ref{th_main}.
It generates the periodic word $(ab)^\omega$. 
This shows the importance of the condition \textit{$f^\omega(a)$ is not periodic} 
that does not occur in Propositions~\ref{sectionCase2} and \ref{sectionCase3},
\end{example}

\begin{example}
\label{example_optimality_3}
Let $\varphi$ be the Fibonacci morphism defined by $\varphi(a) = ab$ and $\varphi(b) = a$.
We have $\varphi^2(a) = aba$ and $\varphi^3(a) = abaab$.
This examples shows the optimality of the exponent 3 in the last condition of Theorem~\ref{th_main}.
\end{example}

\section{Conclusion}

After Theorem~\ref{th_main}, a natural problem is to obtain a characterization of morphisms 
that generate infinite Lyndon words over an alphabet containing at least three letters.

Let us observe that Proposition~\ref{engendreEntrainePresMotsFinis} does not extend to morphisms 
over alphabets with at least three letters. Indeed consider any 
endomorphism $f$ such that $f(a) = a u$ with $u$, $f(b)$ and $f(c)$ belonging to $\{b, c\}^*$ (note that one of the two words $f(b)$ and $f(c)$ could be the empty word: we just need that $\lim_{n \to \infty} |f^n(a)|$ is infinite). 
Then $f^\omega(a)$ is an infinite Lyndon word whatever is $f$ (that may not preserve the lexicographic order). Note that the previous example can include some erasing morphisms. We don't know whether the condition \textit{$f$ preserves the lexicographic order} is necessary if $f$ generates a recurrent word.

Note also that, if an analog of Theorem~\ref{th_main} exists for a larger alphabet $A$, then the exponent in the last condition would be at least 
$\#A+1$ with $\#A$ the cardinality of $A$. 
Indeed if $A_n = \{a_1 \prec \ldots \prec a_n\}$, one can extends Example~\ref{example_optimality_3} defining the morphism $f$ by
$f(a_1) = a_1a_2$, $f(a_i) = a_{i+1}$ for $2 \leq i < n$ and $f(a_n) = a_1$. Then for $1 \leq i \leq n$, $f^i(a_i)$ is a prefix of $a_1a_2\cdots a_n a_1$ and so a prefix of Lyndon word while $f^{n+1}(a)$ is not such a prefix since it begins with $a_1a_2\cdots a_n a_1a_1$.

\bibliographystyle{plain}
\bibliography{local}

\end{document}